\documentclass[10pt,conference]{IEEEtran}

\newtheorem{proposition}{Proposition}
\newtheorem{corollary}{Corollary}
\newtheorem{definition}{Definition}
\newtheorem{lemma}{Lemma}

\newcommand{\df}{\stackrel{\mbox{\scriptsize def}}{=}}
\newcommand{\ELS}{ELS}
\newcommand{\rk}{\mathrm{rk}}

\newcommand{\dr}{d_{\mbox{\tiny{R}}}}
\newcommand{\ds}{d_{\mbox{\tiny{S}}}}
\newcommand{\Ar}{A_{\mbox{\tiny{R}}}}
\newcommand{\As}{A_{\mbox{\tiny{S}}}}
\newcommand{\ar}{a_{\mbox{\tiny{R}}}}
\newcommand{\as}{a_{\mbox{\tiny{S}}}}
\newcommand{\deltar}{\delta_{\mbox{\tiny{R}}}}
\newcommand{\deltas}{\delta_{\mbox{\tiny{S}}}}

\usepackage{amsmath,amssymb,cite}
\usepackage[dvips]{graphicx}

\begin{document}
\title{Constant-Rank Codes}
\author{Maximilien Gadouleau and Zhiyuan Yan\\
Department of Electrical and Computer Engineering \\
Lehigh University, PA 18015, USA\\ E-mails: \{magc,
yan\}@lehigh.edu} \maketitle

\thispagestyle{empty}

\begin{abstract}
Constant-dimension codes have recently received attention due to
their significance to error control in noncoherent random network
coding. In this paper, we show that constant-rank codes are closely
related to constant-dimension codes and we study the properties of
constant-rank codes. We first introduce a relation between vectors
in $\mathrm{GF}(q^m)^n$ and subspaces of $\mathrm{GF}(q)^m$ or
$\mathrm{GF}(q)^n$, and use it to establish a relation between
constant-rank codes and constant-dimension codes. We then derive
bounds on the maximum cardinality of constant-rank codes with given
rank weight and minimum rank distance. Finally, we investigate the
asymptotic behavior of the maximal cardinality of constant-rank
codes with given rank weight and minimum rank distance.
\end{abstract}
\IEEEpeerreviewmaketitle

\section{Introduction}\label{sec:introduction}
While random network coding \cite{HO_IT06} has proved to be a
powerful tool for disseminating information in networks, it is
highly susceptible to errors. Thus, error control for random network
coding is critical and has received growing attention recently.
Error control schemes proposed for random network coding assume two
types of transmission models: some (see, e.g., \cite{yeung_cis06,
cai_cis06}) depend on the underlying network topology or the
particular linear network coding operations performed at various
network nodes; others \cite{koetter_arxiv07, silva_arxiv07} assume
that the transmitter and receiver have no knowledge of such channel
transfer characteristics. The contrast is similar to that between
coherent and noncoherent communication systems.

Error control for noncoherent random network coding is first
considered in \cite{koetter_arxiv07}. Motivated by the property that
random network coding is vector-space preserving,
\cite{koetter_arxiv07} defines an operator channel that captures the
essence of the noncoherent transmission model. Hence, codes defined
in finite field Grassmannians \cite{chihara_siam87}, referred to as
constant-dimension codes, play a significant role in error control
for noncoherent random network coding.
In \cite{koetter_arxiv07}, a Singleton bound for constant-dimension
codes and a family of codes that are nearly Singleton-bound
achieving are proposed.
Despite the asymptotic optimality of the Singleton bound and the
codes designed in \cite{koetter_arxiv07}, the maximal cardinality of
a constant-dimension code with finite dimension and minimum distance
remains unknown, and it is not clear how an optimal code that
achieves the maximal cardinality can be constructed. It is difficult
to answer the above questions based on constant-dimension codes
directly since the set of all subspaces of the ambient space lacks a
natural group structure \cite{silva_arxiv07}.

The class of nearly Singleton bound achieving constant-dimension
codes in \cite{koetter_arxiv07} are related to rank metric codes.
The relevance of rank metric codes to noncoherent random network
coding is further established in \cite{silva_arxiv07}. In addition
to network coding, rank metric codes \cite{delsarte_jct78,
gabidulin_pit0185, roth_it91} have been receiving steady attention
in the literature due to their applications in storage systems
\cite{roth_it91}, public-key cryptosystems \cite{gabidulin_lncs91},
and space-time coding \cite{lusina_it03}. The pioneering works in
\cite{delsarte_jct78, gabidulin_pit0185, roth_it91} have established
many important properties of rank metric codes. Independently in
\cite{delsarte_jct78, gabidulin_pit0185, roth_it91}, a Singleton
bound (up to some variations) on the minimum rank distance of codes
was established, and a class of codes that achieve the bound with
equality was constructed. We refer to codes that attain the
Singleton bound as maximum rank distance (MRD) codes, and the class
of MRD codes proposed in \cite{gabidulin_pit0185} as Gabidulin codes
henceforth.




In this paper, we investigate the properties of constant-rank codes,
which are the counterparts in rank metric codes of constant
(Hamming) weight codes \cite{agrell_it00}.
We first introduce a relation between vectors in
$\mathrm{GF}(q^m)^n$ and subspaces of $\mathrm{GF}(q)^m$ or
$\mathrm{GF}(q)^n$, and use it to establish a relation between
constant-rank codes and constant-dimension codes. We also derive a
lower bound on the maximum cardinality of constant-rank codes which
depends on the maximum cardinality of constant-dimension codes. We
then derive bounds on the maximum cardinality of constant-rank codes
with given rank and minimum rank distance. Finally, we characterize
the asymptotic behavior of the maximal cardinality of constant-rank
codes with given rank and minimum rank distance, and compare it with
asymptotic behavior of the maximal cardinality of constant-dimension
codes.

The rest of the paper is organized as follows.
Section~\ref{sec:preliminaries} briefly reviews some important
concepts in order to keep this paper self-contained. In
Section~\ref{sec:prel_results}, we establish a relation between
constant-dimension and constant-rank codes. In
Section~\ref{sec:bounds}, we derive bounds on the maximum
cardinality of constant-rank codes with a given minimum rank
distance. Finally, Section~\ref{sec:asymptotics} investigates the
asymptotic behavior of the maximum cardinality of constant-rank
codes.

\section{Preliminaries}\label{sec:preliminaries}
\subsection{Rank metric codes and elementary linear subspaces}\label{sec:rank_metric}

Consider a vector ${\bf x}$ of length $n$ over $\mathrm{GF}(q^m)$.
The field $\mathrm{GF}(q^m)$ may be viewed as an $m$-dimensional
vector space over $\mathrm{GF}(q)$. The rank weight of ${\bf x}$,
denoted as $\rk({\bf x})$, is defined to be the \emph{maximum}
number of coordinates of ${\bf x}$ that are linearly independent
over $\mathrm{GF}(q)$ \cite{gabidulin_pit0185}. For any basis $B_m$
of $\mathrm{GF}(q^m)$ over $\mathrm{GF}(q)$, each coordinate of
${\bf x}$ can be expanded to an $m$-dimensional column vector over
$\mathrm{GF}(q)$ with respect to $B_m$. The rank weight of ${\bf x}$
is hence the rank of the $m\times n$ matrix over $\mathrm{GF}(q)$
obtained by expanding all the coordinates of ${\bf x}$.

For all ${\bf x}, {\bf y}\in \mathrm{GF}(q^m)^n$, it is easily
verified that $\dr({\bf x},{\bf y})\df \rk({\bf x} - {\bf y})$ is a
metric over GF$(q^m)^n$, referred to as the \emph{rank metric}
henceforth \cite{gabidulin_pit0185}. The {\em minimum rank distance}
of a code $C$, denoted as $\dr$, is simply the minimum rank distance
over all possible pairs of distinct codewords.

It is shown in \cite{delsarte_jct78, gabidulin_pit0185, roth_it91}
that the minimum rank distance of a block code of length $n$ and
cardinality $M$ over $\mathrm{GF}(q^m)$ satisfies $\dr \leq
n-\log_{q^m}M+1.$ In this paper, we refer to this bound as the
Singleton bound for rank metric codes and codes that attain the
equality as maximum rank distance (MRD) codes. We refer to the
subclass of linear MRD codes introduced independently in
\cite{delsarte_jct78, gabidulin_pit0185, roth_it91} as Gabidulin
codes.

We denote the number of vectors of rank $r$ ($0 \leq r \leq
\min\{m,n\}$) in $\mathrm{GF}(q^m)^n$ as $N_r(q^m,n) = {n \brack r}
\alpha(m,r)$ \cite{gabidulin_pit0185}, where $\alpha(m,0) \df 1$ and
$\alpha(m,r) \df \prod_{i=0}^{r-1}(q^m-q^i)$ for $r \geq 1$. The ${n
\brack r}$ term is often referred to as a Gaussian
polynomial~\cite{andrews_book76}, defined as ${n \brack r} \df
\alpha(n,r)/\alpha(r,r)$. The volume of a ball with rank radius $r$
in $\mathrm{GF}(q^m)^n$ is denoted as $V_r(q^m,n) = \sum_{i=0}^r
N_i(q^m,n)$.
For all $q$, $1 \leq d \leq r \leq n \leq m$, the number of
codewords of rank $r$ in an $(n, n-d+1, d)$ linear MRD code over
$\mathrm{GF}(q^m)$ is given by \cite{gabidulin_pit0185}
\begin{equation}
    \label{eq:Mdr_def}
    M_{d,r} \df {n \brack r} \sum_{j=d}^r (-1)^{r-j}
    {r \brack j} q^{{r-j \choose 2}} \left( q^{m(j-d+1)} - 1\right).
\end{equation}

An {\em elementary linear subspace} (ELS) \cite{gadouleau_it06} is
defined to be a linear subspace $\mathcal{V} \subseteq
\mathrm{GF}(q^m)^n$ for which there exists a basis of vectors in
$\mathrm{GF}(q)^n$. We denote the set of all \ELS{}'s of
$\mathrm{GF}(q^m)^n$ with dimension $v$ as $E_v(q^m,n)$. It can be
easily shown that $|E_v(q^m,n)| = {n \brack v}$ for all $m$. An ELS
has properties similar to those for a set of coordinates
\cite{gadouleau_it06}. In particular, any vector belonging to an
\ELS{} with dimension $r$ has rank no more than $r$; conversely, any
vector ${\bf x} \in \mathrm{GF}(q^m)^n$ with rank $r$ belongs to a
unique \ELS{} in $E_r(q^m,n)$.

\subsection{Constant-dimension codes}

A {\em constant-dimension code} \cite{koetter_arxiv07} of length $n$
and constant-dimension $r$ over $\mathrm{GF}(q)$ is defined to be a
nonempty subset of $E_r(q,n)$. For all $\mathcal{U}, \mathcal{V} \in
E_r(q,n)$, it is easily verified that
\begin{equation}\label{eq:ds}
    \ds(\mathcal{U}, \mathcal{V}) \df \dim(\mathcal{U} + \mathcal{V}) -
    \dim(\mathcal{U} \cap \mathcal{V}) = 2\dim(\mathcal{U} +
    \mathcal{V}) - 2r
\end{equation}
is a metric over $E_r(q,n)$, referred to as the {\em subspace
metric} henceforth \cite{koetter_arxiv07}. The subspace distance
between $\mathcal{U}$ and $\mathcal{V}$ thus satisfies
$\ds(\mathcal{U}, \mathcal{V}) = 2\rk({\bf X}^T \,|\, {\bf Y}^T) -
2r$, where ${\bf X}$ and ${\bf Y}$ are generator matrices of
$\mathcal{U}$ and $\mathcal{V}$, respectively.

The {\em minimum subspace distance} of a constant-dimension code
$\Omega \subseteq E_r(q,n)$, denoted as $\ds$, is the minimum
subspace distance over all possible pairs of distinct subspaces. We
say $\Omega$ is an $(n,\ds,r)$ constant-dimension code over
$\mathrm{GF}(q)$ and we denote the maximum cardinality of an
$(n,2d,r)$ constant-dimension code over $\mathrm{GF}(q)$ as
$\As(q,n,2d,r)$. Since $\As(q,n,2d,r) = \As(q,n,2d,n-r)$
\cite{xia_arxiv07}, only the case where $2r \leq n$ needs to be
considered. Also, since $\As(q,n,2,r) = {n \brack r}$ and
$\As(q,n,2d,r) = 1$ for $d>r$, we shall assume $2 \leq d \leq r$
henceforth. Upper and lower bounds on $\As(q,n,2d,r)$ were derived
in \cite{wang_it03, koetter_arxiv07, xia_arxiv07}. In particular,
for all $q$, $2r \leq n$, and $2 \leq d \leq r$,
\begin{equation}\label{eq:bounds_As}
    q^{(n-r)(r-d+1)} \leq \As(q,n,2d,r) \leq
    \frac{\alpha(n,r-d+1)}{\alpha(r,r-d+1)}.
\end{equation}

\subsection{Preliminary graph-theoretic results}\label{sec:graph_theory}

We review some results in graph theory given in
\cite{el_rouayheb_isit07}. Two adjacent vertices $u,v$ in a graph
are denoted as $u \sim v$.

\begin{definition}\label{def:homo}
Let $G$ and $H$ be two graphs. A mapping $f$ from $V(G)$ to $V(H)$
is a homomorphism if for all $u,v \in V(G)$, $u \sim v \Rightarrow
f(u) \sim f(v)$.
\end{definition}

\begin{definition}\label{def:auto}
Let $G$ be a graph and $\phi$ a bijection from $V(G)$ to itself.
$\phi$ is called an automorphism of $G$ if for all $u,v \in V(G)$,
$u \sim v \Leftrightarrow \phi(u) \sim \phi(v)$.
\end{definition}

\begin{definition}\label{def:vertex_transitive}
We say that the graph $G$ is vertex transitive if  for all $u,v \in
V(G)$, there exists an automorphism $\phi$ of $G$ such that $\phi(u)
= v$.
\end{definition}

An {\em independent set} of a graph $G$ is a subset of $V(G)$ with
no adjacent vertices. The independence number $\alpha(G)$ of $G$ is
the maximum cardinality of an independent set of $G$. If $H$ is a
vertex transitive graph and if there is a homomorphism from $G$ to
$H$, then \cite{el_rouayheb_isit07}
\begin{equation}\label{eq:alpha_G_H}
    \alpha(G) \geq \alpha(H) \frac{|G|}{|H|}.
\end{equation}

\section{Constant-Rank and Constant-Dimension
Codes}\label{sec:prel_results}

\subsection{Definitions and technical results}

\begin{definition}\label{def:constant-rank}
A constant-rank code of length $n$ and constant-rank $r$ over
$\mathrm{GF}(q^m)$ is a nonempty subset of $\mathrm{GF}(q^m)^n$ such
that all elements have rank weight $r$.
\end{definition}

We denote a constant-rank code with length $n$, minimum rank
distance $d$, and constant-rank $r$ as an $(n,d,r)$ constant-rank
code over $\mathrm{GF}(q^m)$. We define the term $\Ar(q^m,n,d,r)$ to
be the maximum cardinality of an $(n,d,r)$ constant-rank code over
$\mathrm{GF}(q^m)$. If $C$ is an $(n,d,r)$ constant-rank code over
$\mathrm{GF}(q^m)$, then the code obtained by transposing all the
expansion matrices of codewords in $C$ forms an $(m,d,r)$
constant-rank code over $\mathrm{GF}(q^n)$ with the same
cardinality. Therefore $\Ar(q^m,n,d,r) = \Ar(q^n,m,d,r)$, and
henceforth we assume $n \leq m$ without loss of generality.

We now define two families of graphs which are instrumental in our
analysis of constant-rank codes.

\begin{definition}\label{def:R_q}
The {\em bilinear forms graph} $R_q(m,n,d)$ has as vertices all the
vectors in $\mathrm{GF}(q^m)^n$ and two vertices ${\bf x}$ and ${\bf
y}$ are adjacent if and only if $\dr({\bf x}, {\bf y}) < d$. The
{\em constant-rank graph} $K_q(m,n,d,r)$ is the subgraph of
$R_q(m,n,d)$ induced by the vectors in $\mathrm{GF}(q^m)^n$ with
rank $r$.
\end{definition}

The orders of the bilinear forms and constant-rank graphs are thus
given by $|R_q(m,n,d)| = q^{mn}$ and $|K_q(m,n,d,r)| = N_r(q^m,n)$.
An independent set of $R_q(m,n,d)$ corresponds to a code with
minimum rank distance $\geq d$. Due to the existence of MRD codes
for all parameter values, we have $\alpha(R_q(m,n,d)) =
q^{m(n-d+1)}$. Similarly, an independent set of $K_q(m,n,d,r)$
corresponds to a constant-rank code with minimum rank distance $\geq
d$, and hence $\alpha(K_q(m,n,d,r)) = \Ar(q^m,n,d,r)$.

\begin{lemma}\label{lemma:vertex_transitive}
The bilinear forms graph $R_q(m,n,d)$ is vertex transitive for all
$q$, $m$, $n$, and $d$. The constant-rank graph $K_q(m,m,d,m)$ is
vertex transitive for all $q$, $m$, and $d$.
\end{lemma}

\begin{proof}
Let ${\bf u}, {\bf v} \in \mathrm{GF}(q^m)^n$. For all ${\bf x} \in
\mathrm{GF}(q^m)^n$, define $\phi({\bf x}) = {\bf x} + {\bf v} -
{\bf u}$. It is easily shown that $\phi$ is a graph automorphism of
$R_q(m,n,d)$ satisfying $\phi({\bf u}) = {\bf v}$. By
Definition~\ref{def:vertex_transitive}, $R_q(m,n,d)$ is hence vertex
transitive.

Let ${\bf u}, {\bf v} \in \mathrm{GF}(q^m)^m$ have rank $m$, and
denote their expansions with respect to a basis $B_m$ of
$\mathrm{GF}(q^m)$ over $\mathrm{GF}(q)$ as ${\bf U}$ and ${\bf V}$,
respectively. For all ${\bf x} \in \mathrm{GF}(q^m)^m$ with rank
$m$, define $\phi({\bf x}) = {\bf y}$ such that ${\bf Y} = {\bf
X}{\bf U}^{-1} {\bf V}$, where ${\bf X}, {\bf Y}$ are the expansions
of ${\bf x}$ and ${\bf y}$ with respect to $B_m$, respectively. We
have $\phi({\bf u}) = {\bf v}$, $\rk(\phi({\bf x})) = m$, and for
all ${\bf x}, {\bf z} \in \mathrm{GF}(q^m)^m$, $\dr(\phi({\bf x}),
\phi({\bf z})) = \rk({\bf X}{\bf U}^{-1} {\bf V} - {\bf Z}{\bf
U}^{-1} {\bf V}) = \rk({\bf X} - {\bf Z}) = \dr({\bf x}, {\bf z})$.
By Definition~\ref{def:auto}, $\phi$ is an automorphism which takes
${\bf u}$ to ${\bf v}$ and hence $K_q(m,m,d,m)$ is vertex
transitive.
\end{proof}

It is worth noting that $K_q(m,n,d,r)$ is not vertex transitive in
general.

\subsection{Constant-dimension and constant-rank
codes}\label{sec:dimension_v_rank}

In \cite{koetter_arxiv07}, constant-dimension codes were constructed
from rank distance codes as follows. Let $C$ be a code with length
$n$ over $\mathrm{GF}(q^m)$. For any ${\bf c} \in C$, consider its
expansion ${\bf C}$ with respect to the basis $B_m$ of
$\mathrm{GF}(q^m)$ over $\mathrm{GF}(q)$, and construct $I({\bf C})
= ({\bf I}_m \,|\, {\bf C}) \in \mathrm{GF}(q)^{m \times m+n}$. Then
$I(C) \df \{ I({\bf C}) | {\bf c} \in C \}$ is a constant-dimension
code in $E_m(q,m+n)$. This relation between rank codes and
constant-dimension codes was also commented in graph-theoretic terms
in \cite{brouwer_book89}.

We introduce a relation between vectors in $\mathrm{GF}(q^m)^n$ and
subspaces of $\mathrm{GF}(q)^m$ or $\mathrm{GF}(q)^n$. For any ${\bf
x} \in \mathrm{GF}(q^m)^n$ with rank $r$, consider the matrix ${\bf
X} \in \mathrm{GF}(q)^{m \times n}$ obtained by expanding all the
coordinates of ${\bf x}$ with respect to a basis $B_m$ of
$\mathrm{GF}(q^m)$ over $\mathrm{GF}(q)$. The column span of ${\bf
X}$, denoted as $\mathfrak{S}({\bf x})$, is an $r$-dimensional
subspace of $\mathrm{GF}(q)^m$, which corresponds to the subspace of
$\mathrm{GF}(q^m)$ spanned by the coordinates of ${\bf x}$. The row
span of ${\bf X}$, denoted as $\mathfrak{T}({\bf x})$, is an
$r$-dimensional subspace of $\mathrm{GF}(q)^n$, which corresponds to
the unique ELS $\mathcal{V} \in E_r(q^m,n)$ such that ${\bf x} \in
\mathcal{V}$.

\begin{lemma}\label{lemma:S_T}
For all $\mathcal{S} \in E_r(q,m)$ and $\mathcal{T} \in E_r(q,n)$,
there exists ${\bf x} \in \mathrm{GF}(q^m)^n$ with rank $r$ such
that $\mathfrak{S}({\bf x}) = \mathcal{S}$ and $\mathfrak{T}({\bf
x}) = \mathcal{T}$.
\end{lemma}

\begin{proof}
Consider the generator matrices ${\bf G} \in \mathrm{GF}(q)^{r
\times m}$ and ${\bf H} \in \mathrm{GF}(q)^{r \times n}$ of
$\mathcal{S}$ and $\mathcal{T}$, respectively. Let ${\bf X} = {\bf
G}^T {\bf H}$ and ${\bf x}$ be the vector whose expansion with
respect to $B_m$ is given by ${\bf X}$. Then $\mathfrak{S}({\bf x})
= \mathcal{S}$ and $\mathfrak{T}({\bf x}) = \mathcal{T}$.
\end{proof}

By Lemma~\ref{lemma:S_T}, the functions $\mathfrak{S}$ and
$\mathfrak{T}$ are surjective. They are not injective, however. For
all $\mathcal{V} \in E_r(q^m,n)$, there exist exactly $\alpha(m,r)$
vectors ${\bf x} \in \mathcal{V}$ with rank $r$
\cite{gadouleau_it06}, hence for all $\mathcal{T} \in E_r(q,n)$
there exist exactly $\alpha(m,r)$ vectors ${\bf x}$ such that
$\mathfrak{T}({\bf x}) = \mathcal{T}$. By transposition, it follows
that there exist exactly $\alpha(n,r)$ vectors ${\bf x}$ such that
$\mathfrak{S}({\bf x}) = \mathcal{S}$ for all $\mathcal{S} \in
E_r(q,m)$.

For any $C \subseteq \mathrm{GF}(q^m)^n$, define $\mathfrak{S}(C)
\df \{ \mathfrak{S}({\bf c})| {\bf c} \in C \}$ and $\mathfrak{T}(C)
\df \{ \mathfrak{T}({\bf c})| {\bf c} \in C \}$. We obtain the
following lemma.

\begin{lemma}\label{lemma:|S(C)|}
For all $C \subseteq \mathrm{GF}(q^m)^n$, we have $|\mathfrak{S}(C)|
\leq |C| \leq \alpha(n,r)|\mathfrak{S}(C)|$ and $|\mathfrak{T}(C)|
\leq |C| \leq \alpha(m,r)|\mathfrak{T}(C)|.$
\end{lemma}


\begin{proposition}\label{prop:subspace_v_rank}
For any constant-dimension code $\Gamma \subseteq E_r(q,m)$, there
exists a constant-rank code $C$ with length $n$ and constant-rank
$r$ over $\mathrm{GF}(q^m)$ such that $r \leq n \leq m$ and
$\mathfrak{S}(C) = \Gamma$. The cardinality of $C$ satisfies
$|\Gamma| \leq |C| \leq \alpha(n,r)|\Gamma|$. On the other hand, for
any constant-dimension code $\Delta \subseteq E_r(q,n)$, there
exists a constant-rank code $D$ with length $n$ and constant-rank
$r$ over $\mathrm{GF}(q^m)$ such that $r \leq n \leq m$ and
$\mathfrak{T}(D) = \Delta$. The cardinality of $D$ satisfies
$|\Delta| \leq |D| \leq \alpha(m,r)|\Delta|$.
\end{proposition}

\begin{proof}
By Lemma~\ref{lemma:S_T}, for any $\mathcal{U} \in \Gamma$ there
exists ${\bf c}_\mathcal{U} \in \mathrm{GF}(q^m)^n$ with rank $r$
such that $\mathfrak{S}({\bf c}_\mathcal{U}) = \mathcal{U}$.
Therefore, the code $C = \{{\bf c}_\mathcal{U}| \mathcal{U} \in
\Gamma\}$ satisfies $\mathfrak{S}(C) = \Gamma$. $C$ is a
constant-rank code with length $n$ and constant-rank $r$ over
$\mathrm{GF}(q^m)$, and by Lemma~\ref{lemma:|S(C)|}, $|C|$ satisfies
$|\Gamma| \leq |C| \leq \alpha(n,r)|\Gamma|$. The proof for $\Delta
\subseteq E_r(q,n)$ is similar and hence omitted.
\end{proof}


Proposition~\ref{prop:subspace_v_rank} shows that constant-dimension
codes can be viewed as a special class of constant-rank codes.
Although the rank metric is not directly related to the subspace
metric in general, the maximal cardinalities of constant-dimension
codes and constant-rank codes are related.

\begin{proposition}\label{prop:A>As}
For all $q$ and $1 \leq r <d \leq n \leq m$,
\begin{equation}\label{eq:A>As}
    \Ar(q^m,n,d,r) \geq \min \{ \As(q,n,2(d-r),r), \As(q,m,2r,r)\}.
\end{equation}
\end{proposition}

\begin{proof}
Let $\Gamma$ be an optimal $(m,2r,r)$ constant-dimension code over
$\mathrm{GF}(q)$ and $\Delta$ be an optimal $(n,2d,r)$
constant-dimension code over $\mathrm{GF}(q)$. Denote their
cardinalities as $\mu = \As(q,m,2r,r)$ and $\nu = \As(q,n,2d,r)$ and
the generator matrices of their component subspaces as $\{{\bf
X}_i\}_{i=0}^{\mu-1}$ and $\{ {\bf Y}_j \}_{j=0}^{\nu-1}$,
respectively. By~(\ref{eq:ds}), for all $0 \leq i < j \leq \nu-1$,
$2\rk({\bf Y}_i^T \,|\, {\bf Y}_j^T) - 2r \geq 2d$, and hence
$\rk({\bf Y}_i^T \,|\, {\bf Y}_j^T) \geq d+r$.

For all $0 \leq i \leq \mu-1$, define ${\bf b}_i = (\beta_{i,0},
\beta_{i,1}, \ldots, \beta_{i,r-1}) \in \mathrm{GF}(q^m)^r$ such
that the expansion of $\beta_{i,l}$ with respect to a basis $B_m$ of
$\mathrm{GF}(q^m)$ is given by the $l$-th row of ${\bf X}_i$. For
all $0 \leq i < j \leq \nu-1$, the matrix $({\bf X}_i^T \,|\, {\bf
X}_j^T)$ has full rank by~(\ref{eq:ds}) and hence the elements
$\{\beta_{i,0}, \ldots, \beta_{i,r-1}, \beta_{j,0}, \ldots,
\beta_{j,r-1} \}$ are linearly independent. We thus define the basis
$\gamma_{i,j} = \{\beta_{i,0}, \ldots, \beta_{i,r-1}, \beta_{j,0},
\ldots, \beta_{j,r-1}, \gamma_{2r}, \ldots, \gamma_{m-1}\}$ of
$\mathrm{GF}(q^m)$ over $\mathrm{GF}(q)$.

We define the code $C \subseteq \mathrm{GF}(q^m)^n$ such that ${\bf
c}_i = {\bf b}_i {\bf Y}_i^T$ for $0 \leq i \leq \min\{\mu,\nu\}-1$.
Expanding ${\bf c}_i$ and ${\bf c}_j$ with respect to the basis
$\gamma_{i,j}$, we obtain $\rk({\bf c}_i) = \rk \left( {\bf Y}_i^T
\,|\, {\bf 0} \right) = r$ and $\dr({\bf c}_i, {\bf c}_j) =\rk
\left({\bf Y}_i^T\,|\, -{\bf Y}_j^T \,|\, {\bf 0} \right) = \rk({\bf
Y}_i^T \,|\, {\bf Y}_j^T) \geq d+r.$ Therefore, $C$ is an
$(n,d+r,r)$ constant-rank code over $\mathrm{GF}(q^m)$ with
cardinality $\min\{\mu,\nu\}$.
\end{proof}

\begin{corollary}\label{cor:A>As}
For all $q$ and $m$,
\begin{eqnarray}
    \label{eq:As_2r}
    \Ar(q^m,n,2r,r) &\geq& \As(q,n,2r,r) \quad \mbox{for }\, n \leq m\\
    \label{eq:As_2(d-r)}
    \Ar(q^m,m,d,r) &\geq& \As(q,m,2r,r)  \quad \mbox{for }\, r <d.
\end{eqnarray}
\end{corollary}


Therefore, a lower bound on $\As$ is also a lower bound on $\Ar$ for
$r < d$. We may use the lower bound on $\As$
in~(\ref{eq:bounds_As}).

\section{Bounds on constant-rank codes}\label{sec:bounds}

We derive bounds on the maximum cardinality of constant-rank codes.
We first observe that $\Ar(q^m,n,d,r)$ is a non-decreasing function
of $m$ and $n$, and a non-increasing function of $d$. We also remark
that the bounds on $\Ar(q^m,n,d,r)$ derived in
Section~\ref{sec:dimension_v_rank} for $2r \leq n$ can be easily
adapted for $2r > n$ by applying them to $n-r$ instead. Finally,
since $\Ar(q^m,n,1,r) = N_r(q^m,n)$ and $\Ar(q^m,n,d,r) = 1$ for $d
> 2r$, we shall assume $2 \leq d \leq 2r$ henceforth.

By considering the Singleton bound for rank metric codes or MRD
codes, we obtain a lower bound and some upper bounds on
$\Ar(q^m,n,d,r)$.

\begin{proposition}\label{prop:A_MRD}
For all $q$ and $1 \leq r,d \leq n \leq m$,
\begin{eqnarray}
    \label{eq:A_lower_MRD}
    \Ar(q^m,n,d,r) &\geq& M_{d,r} \quad \mbox{for }\, r \geq d\\
    \label{eq:A_upper_MRD2}
    \Ar(q^m,n,d,r) &\leq& q^{m(n-d+1)} - \sum_{j \in J_a}
    \Ar(q^m,n,d,j)\\
    \label{eq:A_upper_MRD1}
    \Ar(q^m,n,d,r) &\leq& q^{m(n-d+1)} - \sum_{i \in I_r} M_{d,i}\\
    \label{eq:A_upper_MRD}
    \Ar(q^m,n,d,r) &\leq& q^{m(n-d+1)} - 1 \quad \mbox{for }\, r \geq
    d,
\end{eqnarray}
where $I_r \df \{ i \,:\, 0 \leq i \leq n, |i-r| \geq d \}$ and $J_a
\df I_r \cap \{a + kd \,:\, k \in \mathbb{Z}\}$ for $0 \leq a < d$.
\end{proposition}

\begin{proof}
The codewords of rank $r$ in an $(n, n-d+1,d)$ linear MRD code over
$\mathrm{GF}(q^m)$ form an $(n,d,r)$ constant-rank code. Thus,
$\Ar(q^m,n,d,r) \geq M_{d,r}$ for $r \geq d$.


Let $C$ be an $(n,n-d+1,d)$ linear MRD code over $\mathrm{GF}(q^m)$,
and denote its codewords with ranks belonging to $I_r$ as $C'$. For
$0 \leq j \leq n$, let $C_j$ be optimal $(n,d,j)$ constant-rank
codes and define $C'' \df \bigcup_{j \in J_a} C_j$. The Singleton
bound on the codes $C_r \cup C'$ and $C_r \cup C''$
yields~(\ref{eq:A_upper_MRD1}) and~(\ref{eq:A_upper_MRD2}),
respectively.

Finally, the Singleton bound on $C \cup \{0\}$, where $C$ is an
$(n,d,r)$ ($r \geq d$) constant-rank code over $\mathrm{GF}(q^m)$,
yields~(\ref{eq:A_upper_MRD}).
\end{proof}

\begin{proposition}\label{prop:A_bounds}
For all $q$ and $1 \leq r,d \leq n \leq m$,
\begin{eqnarray}
    \label{eq:bassalygo}
    \Ar(q^m,n,d,r) &\geq& N_r(q^m,n) q^{m(-d+1)}\\
    \nonumber
    \Ar(q^m,m,d,m) &\leq& \Ar(q^{m-1},m-1,d,m-1)\\
    \label{eq:johnson_m-1}
    &\cdot&  q^{m-1}(q^m-1) \quad \mbox{for }\, d<m\\
    \nonumber
    \Ar(q^m,n,d,r) &\leq&  \Ar(q^m,n-1,d,r)\\
    \label{eq:johnson_r}
    &\cdot& \frac{q^n-1}{q^{n-r}-1} \quad \mbox{for }\, r<n.
\end{eqnarray}
\end{proposition}

\begin{proof}
Since $K_q(m,n,d,r)$ is a subgraph of $R_q(m,n,d)$, the inclusion
map is a trivial homomorphism from $K_q(m,n,d,r)$ to $R_q(m,n,d)$.
By Lemma~\ref{lemma:vertex_transitive}, $R_q(m,n,d)$ is vertex
transitive. We hence apply~(\ref{eq:alpha_G_H}) to these graphs,
which yields~(\ref{eq:bassalygo}).

Let $B_{m-1}$ and $B_m$ be bases sets over $\mathrm{GF}(q)$ of
$\mathrm{GF}(q^{m-1})$ and $\mathrm{GF}(q^m)$, respectively. For all
${\bf x} \in \mathrm{GF}(q^{m-1})^{m-1}$ with rank $m-1$, define
$g({\bf x}) = {\bf y} \in \mathrm{GF}(q^m)^m$ such that
\begin{equation}
    \label{eq:g}
    {\bf Y} = \left(\begin{array}{c|c}
    {\bf X} & {\bf 0}\\
    \hline
    {\bf 0} & 1
    \end{array}\right) \in \mathrm{GF}(q)^{m \times m},
\end{equation}
where ${\bf X}$ and ${\bf Y}$ are the expansions of ${\bf x}$ and
${\bf y}$ with respect to $B_{m-1}$ and $B_m$, respectively.
By~(\ref{eq:g}), for all ${\bf x}, {\bf z} \in
\mathrm{GF}(q^{m-1})^{m-1}$ with rank $m-1$, we have $\rk(g({\bf
x})) = \rk({\bf x}) + 1 = m$ and $\rk(g({\bf x}) - g({\bf z})) =
\rk({\bf x} - {\bf z})$. Therefore $g$ is a homomorphism from
$K_q(m-1,m-1,d,m-1)$ to $K_q(m,m,d,m)$.
Applying~(\ref{eq:alpha_G_H}) to these graphs, and noticing that
$\alpha(m,m) = q^{m-1} (q^m-1) \alpha(m-1,m-1)$, we
obtain~(\ref{eq:johnson_m-1}).

We now prove~(\ref{eq:johnson_r}). Note that any vector ${\bf x} \in
\mathrm{GF}(q^m)^n$ with rank $r$ belongs to ${n-r \brack 1}$ ELS's
of dimension $n-1$. Indeed, such ELS's are of the form
$\mathcal{E}({\bf x}) \oplus \mathcal{N}$, where $\mathcal{N} \in
E_{n-r-1}(q^m,n-r)$.

Let $C$ be an optimal $(n,d,r)$ constant-rank code over
$\mathrm{GF}(q^m)$. For all ${\bf c} \in C$ and all $\mathcal{V} \in
E_{n-1}(q^m,n)$, we define $f(\mathcal{V},{\bf c}) = 1$ if ${\bf c}
\in \mathcal{V}$ and $f(\mathcal{V}, {\bf c}) = 0$ otherwise. For
all ${\bf c}$, $\sum_{\mathcal{V} \in E_{n-1}(q^m,n)} f(\mathcal{V},
{\bf c}) = {n-r \brack 1}$, and for all $\mathcal{V}$, $\sum_{{\bf
c} \in C} f(\mathcal{V},{\bf c}) = |
 C \cap \mathcal{V}|$. Summing over all possible pairs, we obtain
\begin{eqnarray*}
    \sum_{\mathcal{V} \in E_{n-1}(q^m,n)} \sum_{{\bf c} \in C} f(\mathcal{V},{\bf
    c}) &=& \sum_{{\bf c} \in C} \sum_{\mathcal{V} \in E_{n-1}(q^m,n)}  f(\mathcal{V},{\bf
    c})\\
    = \sum_{{\bf c} \in C} {n-r \brack 1} &=& {n-r \brack 1} \Ar(q^m,n,d,r).
\end{eqnarray*}
Hence there exists $\mathcal{U} \in E_{n-1}(q^m,n)$ such that $|C
\cap \mathcal{U}| = \sum_{{\bf c} \in C} f(\mathcal{U},{\bf c}) \geq
\frac{{n-r \brack 1}}{{n \brack 1}} \Ar(q^m,n,d,r)$. The restriction
of $C \cap \mathcal{U}$ to the ELS $\mathcal{U}$
\cite{gadouleau_it06} is an $(n-1,d,r)$ constant-rank code over
$\mathrm{GF}(q^m)$, and hence its cardinality satisfies
$\frac{q^{n-r} - 1}{q^n-1} \Ar(q^m,n,d,r) \leq |C \cap \mathcal{U}|
\leq \Ar(q^m,n-1,d,r)$.
\end{proof}

Eq.~(\ref{eq:bassalygo}) is the counterpart in rank metric codes of
the Bassalygo-Elias bound \cite{bassalygo_pit68},
while~(\ref{eq:johnson_r}) is analogous to a well-known result by
Johnson \cite{johnson_it62}. Note that~(\ref{eq:bassalygo}) can be
trivial for $d$ approaching $2r$.

\begin{proposition}\label{prop:A_r_r}
For all $q$ and $1 \leq r \leq n \leq m$,
\begin{equation}\label{eq:A_r_r}
    \Ar(q^m,n,r,r) = {n \brack r}(q^m-1).
\end{equation}
\end{proposition}

\begin{proof}
First, by~(\ref{eq:A_lower_MRD}), we obtain $\Ar(q^m,n,r,r) \geq {n
\brack r}(q^m-1).$ Second, applying~(\ref{eq:johnson_r})
successively $n-r$ times leads to $\Ar(q^m,n,r,r) \leq {n \brack r}
\Ar(q^m,r,r,r).$ By~(\ref{eq:A_upper_MRD}), we obtain
$\Ar(q^m,n,r,r) \leq {n \brack r} (q^m-1)$.
\end{proof}

Equality in~(\ref{eq:A_r_r}) is thus achieved by the codewords of
rank $r$ in an $(n,n-r+1,r)$ linear MRD code.

\begin{proposition}\label{prop:A_upper_n_brack_r}
For all $q$ and $0 \leq r < d \leq n \leq m$,
\begin{equation}\label{eq:A_upper_n_brack_r}
    \Ar(q^m,n,d,r) \leq {n \brack r}.
\end{equation}
\end{proposition}

\begin{proof}
Consider a code $C$ with minimum rank distance $d$ and constant-rank
$r<d$. If $|C| > {n \brack r} = |E_r(q^m,n)|$, then there exist two
codewords in $C$ belonging to the same ELS $\mathcal{V} \in
E_r(q^m,n)$. Their distance is hence at most equal to $r$, which
contradicts the minimum distance of $C$. Therefore, $|C| \leq {n
\brack r}$.
\end{proof}

\begin{corollary}
For all $q$, $m$, and $n$, $\Ar(q^m,n,2,1) = {n \brack 1}$.
\end{corollary}

\begin{proof}
First, by Proposition~\ref{prop:A_upper_n_brack_r}, we obtain
$\Ar(q^m,n,2,1) \leq {n \brack 1}$. Second, by
Corollary~\ref{cor:A>As}, we obtain $\Ar(q^m,n,2,1) \geq
\As(q,n,2,1)$. We now prove that $\As(q,n,2,1) = {n \brack 1}$. For
any $\mathcal{U}, \mathcal{V} \in E_1(q,n)$, $\mathcal{U} \neq
\mathcal{V}$, we have $\dim(\mathcal{U} \cap \mathcal{V}) = 0$ and
hence $\ds(\mathcal{U}, \mathcal{V}) = 2$. Therefore, $E_1(q,n)$ is
a constant-dimension code with minimum subspace distance $2$ and
$\As(q,n,2,1) = {n \brack 1}$.
\end{proof}

\section{Asymptotic results}\label{sec:asymptotics}
In this section, we study the asymptotic behavior of
$\Ar(q^m,n,\dr,r)$. In order to compare it to the asymptotic
behavior of $\As(q,m,\ds,r)$, we use a set of normalized parameters
different from those introduced in \cite{koetter_arxiv07}: $\nu =
\frac{n}{m}$, $\rho = \frac{r}{m}$, $\deltar = \frac{\dr}{m}$, and
$\deltas = \frac{\ds}{2m}$. By definition, $0 \leq \rho, \deltar
\leq \nu$, and since we assume $n \leq m$, $\nu \leq 1$. We consider
the asymptotic rates defined as $\ar(\nu,\deltar,\rho) \df \lim_{m
\rightarrow \infty} \sup \left[\log_{q^{m^2}} \Ar(q^m,n,\dr,r)
\right]$ and $\as(\deltas,\rho) \df \lim_{m \rightarrow \infty} \sup
\left[ \log_{q^{m^2}} \As(q,m,\ds,r) \right].$

Adapting the results in \cite{silva_arxiv07} using the parameters
defined above, we obtain $\as(\deltas, \rho) =
\min\{(1-\rho)(\rho-\deltas),\rho(1-\rho-\deltas)\}$ for $0 \leq
\deltas \leq \min\{\rho,1-\rho\}$ and $\as(\deltas,\rho) = 0$
otherwise.

We now investigate how the $\Ar(q^m,n,d,r)$ term behaves as the
parameters tend to infinity. Without loss of generality, we only
consider the case where $0 \leq \deltar \leq 2\rho$, since
$\ar(\nu,\deltar,\rho) = 0$ for $\deltar
> 2\rho$.

\begin{proposition}\label{prop:a}
Suppose $\nu \leq 1$. For $0 \leq \deltar \leq \rho$,
\begin{equation}\label{eq:a1}
    \ar(\nu,\deltar,\rho) = \rho(1+\nu-\rho) - \deltar.
\end{equation}
For $\rho \leq \deltar \leq \min\{ 2\rho, \nu\}$,
\begin{equation}\label{eq:a2}
    \max\{0, \rho(1+\nu-\rho) - \deltar\}
    \leq \ar(\nu,\deltar,\rho) \leq
    \rho(\nu-\deltar).
\end{equation}

Suppose $\nu > 1$. For $0 \leq \deltar \leq \rho$,
$\ar(\nu,\deltar,\rho) = \rho(1+\nu-\rho) - \nu \deltar.$ For $\rho
\leq \deltar \leq \min\{ 2\rho, 1\}$, $\max\{0, \rho(1+\nu-\rho) -
\nu \deltar\}\leq \ar(\nu,\deltar,\rho) \leq\rho(1-\deltar)$.
\end{proposition}

\begin{proof}
We give the proof for $\nu \leq 1$, and the proof for $\nu > 1$ is
similar and hence omitted. We first derive a lower bound on
$\ar(\nu,\deltar,\rho)$ for all $\rho$. Using the combinatorial
bounds in \cite{gadouleau_it06},~(\ref{eq:bassalygo}) yields
$\Ar(q^m,n,\dr,r)
> q^{r(m+n-r) - \sigma(q) + m(-\dr+1)}$, where $\sigma(q) < 2$ for
$q \geq 2$. This asymptotically becomes $\ar(\nu, \deltar, \rho)
\geq \rho(1+\nu-\rho) - \deltar$ for $0 \leq \deltar \leq \min\{
2\rho, \nu \}$.

We now derive an upper bound on $\ar(\nu,\deltar,\rho)$. First,
suppose $r \geq \dr$. Applying~(\ref{eq:johnson_r}), we easily
obtain $\Ar(q^m,n,\dr,r) \leq {n \brack r} \Ar(q^m,r,\dr,r).$
Combining with~(\ref{eq:A_upper_MRD}), we obtain $\Ar(q^m,n,\dr,r)
\leq {n \brack r} q^{m(r-\dr+1)} < q^{r(n-r) + \sigma(q) +
m(r-\dr+1)}.$ Asymptotically, this becomes $\ar(\nu,\deltar,\rho)
\leq \rho(\nu-\rho) - \deltar + \rho$ for $\rho \geq \deltar$.
Second, suppose $r < \dr$. By the same token, we obtain
$\Ar(q^m,n,\dr,r) \leq \frac{{n \brack r}}{{\dr \brack r}}
\Ar(q^m,\dr,\dr,r) \leq q^{r(n-\dr) +\sigma(q) + m}$, and hence
$\ar(\nu, \deltar, \rho) \leq \rho(\nu-\deltar)$ for $\rho \leq
\deltar$.
\end{proof}


We observe that the asymptotic behavior of the maximal cardinality
of constant-dimension codes depends on whether $\rho =
\frac{r}{m}\leq \frac{1}{2}$, while the asymptotic behavior of the
maximal cardinality of constant-rank codes depends on whether $\nu =
\frac{n}{m}\leq 1$. This is due to the different behaviors of rank
metric codes of length $n$ over $\mathrm{GF}(q^m)$ for $m \geq n$
and $m<n$ respectively. The construction of an asymptotically
optimal constant-dimension code in $E_r(q,m)$ given in
\cite{koetter_arxiv07} and reviewed in
Section~\ref{sec:dimension_v_rank} is based on a rank metric code of
length $m-r$ over $\mathrm{GF}(q^r)$. Hence $r \geq m-r$ for the
rank metric code is equivalent to $r \geq m/2$ (or $\rho \geq 1/2$)
for the constant-dimension code.

By the Singleton bound on rank metric codes, the asymptotic behavior
of the cardinality of an $(n,n-\dr+1,\dr)$ linear MRD code over
$\mathrm{GF}(q^m)$ with $\nu \leq 1$ is given by $\nu-\dr$. However,
by~(\ref{eq:a1}), $\ar(\nu, \deltar, \nu) = \nu-\dr$ for $\nu \leq
1$ and hence the maximum cardinality of a constant-rank code with
rank $n$ is asymptotically equivalent to the cardinality of an MRD
code with the same minimum rank distance.
%
We hence conjecture that the code
formed by the codewords of rank $n$ in an $(n,n-\dr+1,\dr)$ linear
MRD code achieves the maximal cardinality asymptotically.

\bibliographystyle{IEEETran}
\bibliography{gpt,network_coding}

\end{document}